\documentclass[a4paper,10pt,twoside]{article}
\pdfoutput=1

\makeatletter

\usepackage[T1]{fontenc}
\usepackage[utf8]{inputenc}
\usepackage{textcomp}

\usepackage[english]{babel}

\usepackage{enumitem}

\usepackage{color}

\usepackage[intlimits]{amsmath}
\usepackage{amsthm}
\usepackage{mathtools}
\usepackage{tensor}

\usepackage{url}
\usepackage[unicode,colorlinks,pdfusetitle,final]{hyperref}

\usepackage{units}

\usepackage[numbers,sort]{natbib}

\IfFileExists{MinionPro.sty}{
  \usepackage[lf,medfamily,italicgreek,openg]{MinionPro}
}{
  \usepackage{lmodern}
  \usepackage{upgreek}
  \usepackage{amssymb}
  \usepackage{amsfonts}

  \let\smallfrac\tfrac
  \let\upDelta\Delta
  
}

\IfFileExists{MyriadPro.sty}{
  \usepackage[lf,medfamily,onlytext]{MyriadPro}

  \usepackage[sf,bf,small]{titlesec}
  \usepackage[labelfont={sf,bf},labelsep=period]{caption}

  \let\maybesf\sffamily
}{
  \usepackage{helvet}

  \usepackage[bf,small]{titlesec}
  \usepackage[labelfont={bf},labelsep=period]{caption}

  \let\maybesf\rmfamily
}

\usepackage[final]{microtype}

\setlist{noitemsep,topsep=0pt,parsep=0pt,partopsep=0pt}

\titlelabel{\thetitle.\space}

\numberwithin{equation}{section}

\newcounter{and}
\newcommand{\institute}[1]{\newcommand{\@institute}{#1}}
\newcommand{\inst}[1]{\unskip\smash{$^#1$}\setcounter{and}{1}\ignorespaces}
\newcommand{\email}[1]{\href{mailto:#1}{#1}}

\renewcommand{\maketitle}{
  { 
    \raggedright
    \LARGE
    \noindent
    \bfseries
    \maybesf
    \@title
    \par
  }

  \vspace{1.5\baselineskip}

  { 
    \raggedright
    \renewcommand{\and}{
      \unskip, \ignorespaces
    }
    \noindent\ignorespaces\@author\par
  }

  \vspace{0.5\baselineskip}

  { 
    \raggedright
    \small
    \renewcommand{\and}{
      \par\stepcounter{and}
      \hangindent .8em\noindent
      \hbox to .8em{\smash{$^{\arabic{and}}$}}\ignorespaces
    }
    \ifnum\value{and}=0
      \noindent
    \else
      \hangindent .8em\noindent
      \hbox to .8em{\smash{$^{\arabic{and}}$}}\ignorespaces
    \fi
    \ignorespaces\@institute\par
  }
}

\renewenvironment{abstract}{
  \addvspace{1.5\baselineskip}
  \topsep=0pt\partopsep=0pt
  \trivlist\item[\hskip\labelsep\bfseries\maybesf Abstract.]
}{}

\newenvironment{acknowledgments}{
  \addvspace{1.5\baselineskip}
  \topsep=0pt\partopsep=0pt
  \trivlist\item[\hskip\labelsep\bfseries\maybesf Acknowledgments.]
}{}

\newcommand{\latinabbr}{\textit}
\newcommand{\eg}{\latinabbr{e.g.}\ }

\newcommand{\ie}{\latinabbr{i.e.}}
\newcommand{\cf}{\latinabbr{c.f.}\ }

\newcommand{\kk}{{\vec{k}}}
\newcommand{\pp}{{\vec{p}}}
\newcommand{\xx}{{\vec{x}}}

\newcommand{\defn}{=} 

\newcommand{\e}{\mathrm{e}}
\newcommand{\im}{\mathrm{i}}

\newcommand{\field}[1][K]{\mathbb{#1}}
\newcommand{\RR}{\field[R]}

\DeclareMathOperator{\arccoth}{arccoth}
\newcommand{\Symm}{\mathfrak{S}}

\DeclarePairedDelimiter{\abs}{\lvert}{\rvert}

\newcommand{\order}[1]{\mathrm{O}(#1)}
\DeclarePairedDelimiter{\E}{\langle}{\rangle} 
\newcommand{\norder}[1]{\mathord{:}{#1}\mathord{:}} 

\newcommand{\dif}{\mathrm{d}}
\newcommand{\od}[3][]{\frac{\dif^{#1}#2}{\dif#3^{#1}}}
\newcommand{\pd}[3][]{\frac{\partial^{#1}#2}{\partial#3^{#1}}}
\newcommand{\laplace}{\upDelta} 

\newcommand{\FF}{{}_2F_1}

\newtheoremstyle{maybesf}{}{}{\itshape}{}{\bfseries\maybesf}{.}{\labelsep}{}

\theoremstyle{maybesf}
\newtheorem{theorem}{Theorem}[section]
\newtheorem{lemma}[theorem]{Lemma}
\newtheorem{proposition}[theorem]{Proposition}

\definecolor{hypercolor}{rgb}{0,0.2,0.7}
\hypersetup{
  linkcolor=hypercolor,
  urlcolor=hypercolor,
  citecolor=hypercolor,
  pdfauthor=Daniel Siemssen,
  pdfcreator=LaTeX,
  pdfproducer=pdfTeX
}

\newcommand{\Hadamard}{\mathcal{H}}
\newcommand{\KG}{\mathrm{P}}            
\newcommand{\MM}{\mathbb{M}}            

\makeatother

\begin{document}

\title{Scale-Invariant Curvature Fluctuations from an Extended Semiclassical Gravity}

\author{
  Nicola Pinamonti\inst{1}\inst{2}
	\and
  Daniel Siemssen\inst{1}
}

\institute{
  Dipartimento di Matematica, Università di Genova, Via Dodecaneso 35,
  16146 Genova, Italy.
	\and
  INFN Sezione di Genova, Via Dodecaneso 33,
  16146 Genova, Italy.
	\\
  E-Mail: \email{pinamont@dima.unige.it}, \email{siemssen@dima.unige.it}
}

\maketitle

\begin{abstract}
  We present an extension of the semiclassical Einstein equations which couples $n$-point correlation functions of a stochastic Einstein tensor to the $n$-point functions of the quantum stress-energy tensor.
  We apply this extension to calculate the quantum fluctuations during an inflationary period, where we take as a model a massive conformally coupled scalar field on a perturbed de Sitter space and describe how a renormalization independent, almost-scale-invariant power spectrum of the scalar metric perturbation is produced.
  Furthermore, we discuss how this model yields a natural basis for the calculation of non-Gaussianities of the considered metric fluctuations.
\end{abstract}


\section{Introduction}
\label{sec:introduction}


Shortly after the discovery of the cosmic microwave background (CMB) by Penzias and Wilson \cite{penzias:1965}, Sachs and Wolfe predicted anisotropies in the angular temperature distribution \cite{sachs:1967}.
In their famous paper they discuss what was later coined the \emph{Sachs--Wolfe effect}: The redshift in the microwave radiation caused by fluctuations in the gravitational field and the corresponding matter density fluctuations.
In the standard model of inflationary cosmology the fluctuations imprinted upon the CMB are seeded by quantum fluctuations during inflation \cite{mukhanov:1981,mukhanov:1992,durrer:2008,ellis:2012}.

The usual computation of the power spectrum of the initial fluctuations produced by single-field inflation can be sketched as follows \cite{bartolo:2004,durrer:2008,ellis:2012}:
First, one introduces a (perturbed) classical scalar field $\varphi + \delta \varphi$, the inflaton field, which is coupled to a (perturbed) expanding spacetime $g + \delta g$.
Then, taking the Einstein equation and the Klein--Gordon equation at first order in the perturbation variables, one constructs an equation of motion for the \emph{Mukhanov--Sasaki variable} $Q = \delta \varphi + \dot\varphi\, H^{-1} \Phi$, where $\Phi$ is the \emph{Bardeen potential} \cite{bardeen:1980} and $H$ the Hubble constant.
$Q$ is then quantized\footnote{A recent discussion about the quantization of a such system can be found in \cite{eltzner:2013}.} (in the slow-roll approximation) and one chooses as the state of the associated quantum field a Bunch--Davies-like state.
Last, one evaluates the power spectrum $P_Q(k)$ of $Q$, \ie, the Fourier-transformed two-point distribution of the quantum state, in the super-Hubble regime $k \ll aH$ and obtains an expression of the form\footnote{An alternative definition of the power spectrum is $\mathcal{P}_Q(k) = (2 \uppi)^{-2} k^3 P_Q(k)$.}
\begin{equation}
  \label{eq:powerspectrum_Q}
  P_Q(k) = \frac{A_Q}{k^3}\, \left( \frac{k}{k_0} \right)^{n_s-1},
\end{equation}
where $A_Q$ is the amplitude of the fluctuations, $k_0$ a pivot scale and $n_s$ the spectral index.
Notice the factor of $k^{-3}$ in \eqref{eq:powerspectrum_Q} which gives the spectrum the `scale-invariant' \emph{Harrison--Zel'dovich} form if $n_s = 1$.
Depending on the details of model, $n_s \lesssim 1$ and there is also a possibility for a scale dependence of $n_s$ -- the `running' of the spectral index $n_s = n_s(k)$.

This result can then be related to the power spectrum of the \emph{comoving curvature perturbation} $\mathcal{R}$, which is proportional to $Q$, and can be compared with observational data.
Assuming adiabatic and Gaussian initial perturbations, the WMAP collaboration finds $n_s = 0.9608 \pm 0.0080$ (at $k_0 = \unit[0.002]{Mpc^{-1}}$) in a model without running spectral index and gravitational waves, excluding a scale-invariant spectrum at $5\sigma$ \cite{hinshaw:2012}.
Furthermore, the data of WMAP and other experiments can be used to constrain the deviations from a pure Gaussian spectrum, the so called non-Gaussianities, that arise in some inflationary models \cite{bartolo:2004,bennett:2012,maldacena:2003}.

In \cite{parker:2007,agullo:2008,agullo:2009,agullo:2011a} concerns have been raised whether the calculation leading to \eqref{eq:powerspectrum_Q} and similar calculations are correct: The authors argue that the two-point distribution of the curvature fluctuations has to be regularized and renormalized similarly to what is done in semiclassical gravity.
As a result the power spectrum is changed sufficiently that previously observationally excluded inflation models become realistic again.
On the contrary the authors of \cite{durrer:2009,marozzi:2011} argue that the adiabatic regularization employed in \cite{parker:2007,agullo:2008,agullo:2009,agullo:2011a} is not appropriate for low momentum modes if evaluated at the Horizon crossing and irrelevant for these modes if evaluated at the end of inflation.
We will come back to the issue of renormalization later in this paper and find that in our toy model the considered fluctuations must not be renormalized.

A slightly different approach to the calculation of the power spectrum based on \emph{stochastic gravity} can be found in \cite{roura:2000,roura:2008,hu:2008}.
In spirit similar to the approach presented in this paper, the authors equate fluctuations of the stress-energy tensor with the correlation function of the Bardeen potential.
In the super-Hubble regime they obtain an almost scale-invariant power spectrum.
Moreover, they discuss the equivalence of their stochastic gravity approach with the usual approach of quantizing metric perturbations.

Inspired by these works we follow an approach strictly different from the standard one.
Instead of quantizing a coupled system of linear inflaton and gravitational perturbations, we aim at extending the \emph{semiclassical Einstein equation} to describe metric fluctuations via the fluctuations in the stress-energy tensor of a quantum field.
More precisely, we choose as quantum matter a massive scalar quantum field and analyze carefully the expectation values of the corresponding stress-energy tensor and of their products over a Gaussian Hadamard state.
We then interpret the Einstein tensor as a stochastic field and equate its $n$-point function with the symmetrized $n$-point function of the quantum stress-energy tensor.
Subsequently we use the obtained system to discuss the influence of quantum fluctuations on metric perturbations over the flat patch of a de Sitter background.
We prove that in the limit of early time the power spectrum of these perturbations tends to the Harrison--Zel'dovich spectrum.
The results show that no special features of the metric nor any renormalization ambiguities enter the computation\footnote{In \cite{ford:2010} it was previously noted that fluctuations passively induced by purely conformal matter are not affected by renormalization ambiguity.}; only the Hadamard form of the underlying state matters.
Furthermore, contrary to what is found \eg in \cite{ford:2010}, the obtained power spectrum is not sensitive to the initial time, which is set to be at past infinity in this work.
This initial time independent power spectrum is in agreement with the non-perturbative analysis presented in \cite{frob:2012,frob:2013}.
Finally, we demonstrate that, despite the linearity of the quantum field, the three-point function of the considered metric perturbation does not vanish.

Although a complete theory of perturbations around semiclassical equations has not yet been developed
and although the conjectured equation among matter--gravity fluctuations is not supported by any formal derivation, the preliminary analysis presented in the present paper furnishes an almost scale-invariant power spectrum (up to very low frequencies) for certain metric perturbations very similar to the power spectrum that arises in the standard analysis based on the quantization of the Mukhanov--Sasaki variable evaluated in the Bunch--Davies state.
The main conceptual difference, is that in the extended semiclassical picture, there is less freedom in the choice of the quantum state for the passively induced metric perturbations. To a large extent, the quantum state has already been fixed by the requirement that the background metric satisfies the semiclassical equation.

The paper is organized as follows:
In the next section we shall briefly discuss the form of the matter quantum field we shall use and its backreaction on the background by means of semiclassical Einstein equation.
The basic idea about the influence of quantum matter fluctuations on metric perturbations is outlined in the third section.
The fourth section contains the construction of a simple model on de Sitter spacetime, the analysis of the power spectrum of the considered fluctuations and a brief discussion on the naturally obtained non-Gaussianities.
Finally, the last section, contains a discussion and a summary of the obtained results.


\section{Review: Semiclassical Einstein equation and scalar fields}

The semiclassical Einstein equation\footnote{We use units where $c = \hbar = 8 \uppi\, \mathrm{G} = 1$.}
\begin{equation}
  \label{eq:einstein}
  G_{ab} = \omega(\norder{T_{ab}})
\end{equation}
for a quantum state $\omega$ and a (conserved) quantum stress-energy tensor $T_{ab}$ on a globally hyperbolic spacetime $(M, g)$ allows us to analyze the backreaction of quantum fields on curvature, see \eg \cite{wald:1977,dappiaggi:2008,starobinsky:1980}.

In this paper, for simplicity, we will take as the quantum field a free \emph{conformally coupled scalar field} $\varphi$ with mass $m$, \ie, the classical dynamics are governed by the \emph{Klein--Gordon equation} $\KG \varphi \defn -\Box \varphi +\, \smallfrac{1}{6} R\, \varphi + m^2 \varphi = 0$.
The stress-energy tensor associated to this scalar field can be written as
\begin{equation*}
  T_{ab} = \frac{2}{3} (\partial_a \varphi) (\partial_b \varphi) - \frac{1}{6} g_{ab} (\partial_c \varphi) (\partial^c \varphi) - \frac{1}{6} \varphi \left( m^2 g_{ab} + \frac{R}{6} g_{ab} - R_{ab} + 2 \nabla_a \partial_b \right) \varphi,
\end{equation*}
which differs from the standard definition by the term $\smallfrac{1}{3} g_{ab}\, \varphi\, \KG \varphi$.
This term was introduced in \cite{moretti:2003} to account for the trace anomaly in the quantum theory.
A quick calculation gives us the trace of $T_{ab}$ as
\begin{equation}
  \label{eq:T_trace}
  T \defn g^{ab} T_{ab} = - m^2 \varphi^2 + \frac{1}{3} \varphi\, \KG \varphi.
\end{equation}

The quantization of scalar fields on globally hyperbolic manifolds in the algebraic approach has been discussed thorougly in the existing literature, see \eg \cite{brunetti:2003}, and will not be repeated here.
To accomplish the quantization we shall assume that the quantum state $\omega$ is a
\emph{quasi-free state},\footnote{A quasi-free state is also called \emph{Gaussian state}.} hence all its $n$-point distributions $\omega_n$ descends from its two-point distribution $\omega_2$.
Furthermore, to make \eqref{eq:einstein} well-defined, we will require that $\omega_2$ satisfies the \emph{microlocal spectrum condition} which is equivalent to it being of \emph{Hadamard form} \cite{radzikowski:1996,brunetti:1996,sahlmann:2001}, \ie, if $x,y \in M$ are in a geodesically convex neighbourhood, $\omega_2(x,y)$ is of the form
\begin{equation}
  \label{eq:hadamard}
  \omega_2 = \lim_{\varepsilon \to 0^+} \left( \frac{U}{\sigma_\varepsilon} + V \ln {\frac{\sigma_\varepsilon}{\lambda}} \right) + W = \Hadamard + W,
\end{equation}
where $\sigma_\varepsilon(x,y) = \sigma(x,y) - \im\varepsilon\,\big(t(x)-t(y)\big) + \varepsilon^2$ with $\sigma$ being the squared signed geodesic distance and $t$ a time function, $U(x,y)$ and $V(x,y) = \sum_n V_n(x,y)\, \sigma(x,y)^n$ are smooth functions depending only on the local geometry and $W(x,y)$ is a smooth, symmetric function which characterizes the state.

The distribution $\Hadamard$ in \eqref{eq:hadamard} is called the \emph{Hadamard singularity} and leads to the singular UV behaviour in the quantum theory.
When we regularize $\omega$, by point-splitting regularization, to yield a finite Einstein tensor on the left-hand side of \eqref{eq:einstein}, $\Hadamard$ is subtracted from $\omega_2$.
This regularization is indicated by the normal ordering $\norder{\,\cdot\,}$ in \eqref{eq:einstein}, \ie, in the semiclassical Einstein equation we are evaluating \emph{Wick polynomials}.
The ambiguity in the normal ordering prescription leading to the \emph{renormalization freedom} has been rigorously analyzed in \cite{hollands:2001,hollands:2005}, where it was shown that in a locally covariant quantum field theory the ambiguity amounts to adding polynomials in local curvature and the mass parameter.

We thus note the following regarding the trace of the semiclassical Einstein equation \eqref{eq:einstein}:
Selecting a Hadamard state $\omega$ for the scalar field $\varphi$ on the globally hyperbolic spacetime $(M,g)$, we obtain (see also \cite{hollands:2001,dappiaggi:2008})
\begin{equation*}
  - R = \omega(\norder{T}) = - m^2 [W] + 2 [V_1] + \alpha\, m^4 + \beta\, m^2 R + \gamma\, \Box R,
\end{equation*}
where $[\cdot]$ denotes the Synge bracket $[f](x) = f(x,x)$, $[W]$ is a state dependent contribution, $[V_1]$ is due to the trace anomaly \cite{wald:1978}, and $\alpha, \beta, \gamma$ are renormalization constants.
Since we want to have a well-posed initial value problem, according to \cite{wald:1977}, we shall choose $\gamma = 0$.
Furthermore, different choices of $\beta$ in the equation $-R = \omega(\norder{T})$, can be reabsorbed in the redefinition of the Newton constant.
Because we do not want to change it, we shall fix $\beta = 0$.
With these choices we notice that the single remaining freedom is in the constant $\alpha$ whose value also depends on the choice of the scale $\lambda$ in the Hadamard singularity \eqref{eq:hadamard}.


\section{Higher moments of the semiclassical Einstein equation}

As noted in the introduction above, in the semiclassical approach we are equating a classical quantity, the Einstein tensor, with the expectation value of a quantum observable, the quantum stress-energy tensor, \ie, a probabilistic quantity.
Such a system could make sense only when the fluctuations of the quantum stress-energy tensor can be neglected.
Unfortunately, as also noticed in \cite{pinamonti:2011}, the variance of quantum unsmeared stress-energy tensor is always divergent even when proper regularization methods are considered.
\footnote{\textbf{We recall that, although $\norder{T_{ab} T_{cd}}$ (the normally ordered product of two stress-energy tensors) is a well-defined field, the variance of $\norder{T_{ab}}$ corresponds to $\norder{T_{ab}}\norder{T_{cd}}$ (the product of two normally ordered stress-energy tensors), which is a divergent quantity.}} 
The situation is slightly better when a smeared stress-energy tensor is analyzed.
In that way, however, the covariance of \eqref{eq:einstein} gets lost.
A possible way out is to allow for random perturbations also on the left hand side of \eqref{eq:einstein}.
This is the point of view we shall assume in this paper.
Notice that this is similar in spirit with what happens in the analysis of the random forces leading to Brownian motions by means of Langevin equations.

Consider now the Einstein tensor as a random field.
Then we could imagine to equate the probability distribution of the Einstein tensor with the probability distribution of the stress-energy tensor.
This suggestion, however, seems largely void without a possibility of actually computing the probability distributions of the stress-energy tensor because, as discussed above, its moments of order larger than one are divergent.

Instead we may approach this idea by equating the hierarchy of $n$-point functions of the Einstein tensor with that of the stress-energy tensor:
\begin{subequations}
  \label{eq:einstein_prob}
  \begin{align}
    \label{eq:einstein_prob_1}
    \E[\big]{G_{ab}(x_1)} & = \omega\big( \norder{T_{ab}(x_1)} \big), \\
    \label{eq:einstein_prob_2}
    \E[\big]{\delta G_{ab}(x_1)\, \delta G_{c'd'}(x_2)} & = \frac{1}{2}\, \omega\big( \norder{\delta T_{ab}(x_1)}\, \norder{\delta T_{c'd'}(x_2)} + \norder{\delta T_{c'd'}(x_2)}\, \norder{\delta T_{ab}(x_1)} \big), \\
    \notag
    & \mathrel{\makebox[\widthof{=}]{\vdots}} \\
    \label{eq:einstein_prob_n}
    \E[\big]{(\delta G)^{\boxtimes n}(x_1, \ldots, x_n)} & = \Symm \Big[ \omega\big(\norder{\delta T}^{\boxtimes n}(x_1, \ldots, x_n) \big) \Big], \quad n > 1,
  \end{align}
\end{subequations}
where $\omega$ is a Hadamard state and we defined $\delta G_{ab} \defn G_{ab} - \E{G_{ab}}$ and $\norder{\delta T_{ab}} \defn \norder{T_{ab}} - \omega(\norder{T_{ab}})$.\footnote{The symbol $\boxtimes$ stands for the tensor product at different points, the so called exterior tensor product.}
The symmetrization on the right-hand side, which we denote by $\Symm$, is necessary because the classical quantity on left-hand side is invariant under permutation of the points $x_i$.

We emphasize that we are equating singular functions in \eqref{eq:einstein_prob_n}.
Having all the $n$-point functions of the Einstein stochastic tensor, we can easily construct an equation for the moments of the smeared Einstein tensor which equals the moments of a smeared stress-energy tensor by smearing both sides of \eqref{eq:einstein_prob} with tensor products of a smooth compactly supported function.
This smearing also automatically accounts for the symmetrization in \eqref{eq:einstein_prob}.

Furthermore we stress that equating moments, obtained smearing both side of \eqref{eq:einstein_prob}, is not equivalent to equating probability distributions.
Although it is also possible to arrive at a description in terms of moments when coming from a probability distribution, the inverse mapping is not necessarily well-defined.
Successful attempts to construct a probability distribution for smeared stress-energy tensors can be found in \cite{fewster:2010,fewster:2012b}.

Consider now a Gaussian Hadamard state $\omega$ of a conformally coupled scalar field $\varphi$ on a spacetime $(M, \overline g)$, the \emph{background spacetime}.
Our aim is to calculate the perturbation of the background spacetime as specified by the
correlation functions
on the left-hand side of \eqref{eq:einstein_prob} due to the fluctuations of the stress-energy in the quantum state $\omega$ as specified on the right-hand side of \eqref{eq:einstein_prob}.
In particular we will require that $\omega$ satisfies \eqref{eq:einstein_prob_1} when we identify the Einstein tensor of the background spacetime $\overline G_{ab}$ with $\E{G_{ab}}$ (\cf \cite{pinamonti:2011,pinamonti:2013a} for a discussion of the solutions of the semiclassical Einstein equation in cosmological spacetimes).
Note that by choosing this Ansatz we are completely ignoring any backreaction effects of the fluctuations to the background metric and evaluate the stress-energy tensor on a state specified on the background spacetime.

Later on we consider perturbations of the scalar curvature induced by a `Newtonianly' perturbed FLRW metric.
For this reason it will be sufficient to work with the trace of \eqref{eq:einstein_prob} (using the background metric) instead of the full equations.
With the definition
\begin{equation*}
  S \defn - \overline g^{ab} G_{ab},
\end{equation*}
such that $R = \E{S}$, the equations \eqref{eq:einstein_prob} simplify to
\begin{subequations}
  \label{eq:einstein_prob_trace}
  \begin{align}
    \label{eq:einstein_prob_trace_1}
    \E[\big]{S(x_1)} & = m^2 [W] - 2 [V_1] - \alpha\, m^4, \\
    \label{eq:einstein_prob_trace_2}
    \E[\big]{S(x_1)\, S(x_2)} - \E[\big]{S(x_1)}\E[\big]{S(x_2)} & = m^4 \big( \omega_2^2(x_1, x_2) + \omega_2^2(x_2, x_1) \big), \\
    \notag
    & \mathrel{\makebox[\widthof{=}]{\vdots}} \\
    \label{eq:einstein_prob_trace_n}
    \E[\big]{(S - \E{S})^{\boxtimes n}(x_1, \ldots, x_n)} & = 2^n m^{2n}\, \Symm \left[ \sum_G \prod_{i, j}\frac{\omega_2^{\lambda^G_{ij}}(x_i, x_j)}{\lambda^G_{ij}!} \right], \quad n > 1,
  \end{align}
\end{subequations}
where the sum is over all directed graphs $G$ with $n$ vertices $1, \ldots, n$ with two arrows at every vertex directed to a vertex with a larger label.
$\lambda^G_{ij} \in \{0,1,2\}$ is the number of arrows from $i$ to $j$.
If we perform the symmetrization in \eqref{eq:einstein_prob_trace_n}, we see that the sum is over all acyclical directed graphs with two arrows at every vertex.
For illustration some graphs are shown in Fig. \ref{fig:graphs}.

\begin{figure}
  \centering
	\includegraphics{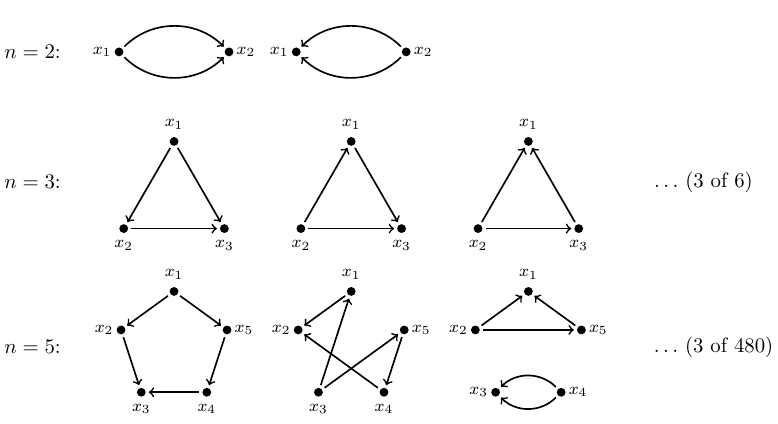}
  \caption{A few graphs illustrating \eqref{eq:einstein_prob_trace_n} for $n=2$, $n=3$ and $n=5$.}
  \label{fig:graphs}
\end{figure}

To obtain \eqref{eq:einstein_prob_trace_2} and \eqref{eq:einstein_prob_trace_n}, note that $\norder{\varphi^2} - \omega(\norder{\varphi^2})$ doesn't depend on the choice of normal ordering (indeed this holds true if we replace $\varphi^2$ with $L \varphi^2$, for any linear operator $L$), and thus only \eqref{eq:einstein_prob_trace_1} needs to be renormalized.
Therefore we may choose normal ordering with respect to $\omega_2$ to see that the combinatorics are equivalent to those in Minkowski space.
Moreover, as $\omega_2$ is a bisolution of the Klein--Gordon equation, the term $\smallfrac{1}{3} \varphi\, \KG \varphi$ which causes the conformal anomaly in \eqref{eq:einstein_prob_trace_1} does not contribute to the higher moments.


\section{Fluctuations around a de Sitter spacetime}
\label{sec:fluctuations_around_de_sitter}

We shall now specialize the general discussion presented above to a special class of fluctuations around an exponentially expanding, flat FLRW universe; the perturbed spacetime is called a \emph{Newtonianly} perturbed FLRW spacetime in \cite{ishibashi:2006}.
That is, the background spacetime $(M, \overline g)$ is given in conformal time $\tau < 0$ by the metric tensor
\begin{equation*}
  \label{eq:background_metric}
  \overline g \defn (H \tau)^{-2} (-\dif\tau \otimes \dif\tau + \delta_{ij}\; \dif x^i \otimes \dif x^j)
\end{equation*}
and we consider fluctuations of the scalar curvature derived from metric perturbations of the form
\begin{equation}
  \label{eq:perturbed_metric}
  g \defn (H \tau)^{-2} \big( -(1 + 2 \Psi)\, \dif\tau \otimes \dif\tau + (1 - 2 \Psi)\, \delta_{ij}\, \dif x^i \otimes \dif x^j \big).
\end{equation}
The kind of fluctuations that we consider by choosing \eqref{eq:perturbed_metric} resemble those that are present in single-scalar field inflation in the longitudinal gauge, where there are only `scalar fluctuations' without anisotropic stress (so that the two Bardeen potentials coincide) \cite{ellis:2012,mukhanov:1992}.
Notice that, for classical metric perturbation, these constraints descend from the linearized Einstein equation, however, \latinabbr{a priori} there is no similar constraint in \eqref{eq:einstein_prob_2}.
Despite these facts, we proceed analyzing the influence of quantum matter on this special kind of metric perturbations and we also refrain from discussing the gauge problem associated to choosing a perturbed spactime; the chosen perturbation potential $\Psi$ is \emph{not} gauge invariant.

We can now calculate the various perturbed curvature tensors and obtain in particular
\begin{equation*}
  S = 12 H^2 (1 - 3 \Psi) + 24 H^2 \tau\, \Psi' - 6 H^2 \tau^2 \Psi'' + 2 H^2 \tau^2\, \laplace \Psi + \order{\Psi^2}
\end{equation*}
for the trace of the perturbed Einstein tensor, where $\laplace$ is the standard Laplace operator.
Dropping terms of higher than linear order, this can also be written as
\begin{equation}
  \label{eq:hyp_op}
  S - \E{S} = - 6 H^2 \tau^4 \Bigg( \pd[2]{}{\tau} - \frac{1}{3} \laplace \Bigg)\, \tau^{-2} \Psi,
\end{equation}
where $\E{S} = 12 H^2$ is nothing but the scalar curvature of the background spacetime.
Notice that, up to a rescaling, the operator on the right hand side of \eqref{eq:hyp_op} looks like a wave operator with the characteristic velocity equal to $1/\sqrt{3}$ of the velocity of light.

We can now evaluate the influence of quantum matter fluctuations on the metric fluctuations by inverting the previous hyperbolic operator by means of its retarded fundamental solutions $\Delta_R$ and applying it on both sides of \eqref{eq:einstein_prob_trace_2} and \eqref{eq:einstein_prob_trace_n}.
From \eqref{eq:einstein_prob_trace_2} we can then (formally) obtain the two-point correlation functions of $\Psi$ (per definition $\E{\Psi} = 0$):
\begin{equation}
  \label{eq:retarded_applied_to_w}
  \E[\big]{\Psi(x_1)\, \Psi(x_2)} = m^4 \iint_{\RR^8} \Delta_R(x_1, y_1)\, \Delta_R(x_2, y_2) \big( \omega_2^2(y_1, y_2) + \omega_2^2(y_2, y_1) \big)\, \dif^4 y_1\, \dif^4 y_2.
\end{equation}
Employing the retarded fundamental solutions in the inversion without adding any solution of \eqref{eq:hyp_op}, we are implicitly assuming that all the $n$-point functions of the perturbation potential $\Psi$ are sourced by quantum fluctuations.
Here we are only interested in evaluating their effect.

\subsection{Bunch--Davis state and the squared two-point distribution}

In order to proceed with our analysis, we shall specify the quantum state $\omega$ for the matter theory.
Following the Ansatz discussed in the preceding section, we choose a quasi-free Hadamard state which satisfies the semiclassical Einstein equation on the background.
In particular, we require that $\omega$ solves \eqref{eq:einstein_prob_trace_1}, namely
\begin{equation*}
  12 H^2 = m^2 [W] - 2 [V_1] - \alpha\, m^4.
\end{equation*}
The right hand side of the previous equation is characterized by three contributions:
The state dependent part $m^2 [W]$, the anomaly part $2 [V_1]$, which takes a very simple form and is proportional to $H^4$, and the renormalization freedom $\alpha m^4$ (remember that we fixed the other renormalization constants to zero).
For the semiclassical Einstein equation to hold, we therefore have to require that $[W]$ is a constant.
Then, having fixed $H$ and $m$ (no matter their absolute value), there is always a choice of $\alpha$ for which the chosen metric $\overline{g}$ and $\omega$ satisfy the semiclassical Einstein equation.
On a de Sitter spacetime these criteria are satisfied by the well known Bunch--Davis state $\omega_{BD}$ \cite{bunch:1978} -- the only Hadamard state which is invariant under the symmetry group of de Sitter space.

In order to evaluate the influence of the quantum matter fluctuations on $\Psi$ via equation \eqref{eq:retarded_applied_to_w}, we have to discuss the form of the two-point distribution of the chosen state and its square.
The two-point distribution of the Bunch--Davis vacuum for a massive, conformally coupled field on de Sitter spacetime, takes the well known expression \cite{bunch:1978,schomblond:1976,allen:1985,kirsten:1993}
\begin{equation*}
  \label{eq:bunch_davies_1}
  \omega_{BD} \defn
  \frac{m^2}{16 \uppi \cos(\uppi\, \nu)}\,
  \FF\left(\frac{3}{2}+\nu, \frac{3}{2}-\nu; 2; \frac{1+Z}{2}\right),
  \quad
  \nu \defn \sqrt{\frac{1}{4} - \frac{m^2}{H^2}},
\end{equation*}
where $\FF$ is the ordinary hypergeometric function and the necessary $\varepsilon$-prescription has been omitted.
Moreover, the auxiliary function $Z(x_1, x_2)$ appearing in the preceding formula is nothing but the geodesic distance in the five dimensional Minkowski space into which de Sitter space can be embedded as an hyperboloid.
Omitting again the necessary $\varepsilon$-prescription, $\omega_{BD}$ can also be recast into
\begin{equation*}
  \label{eq:bunch_davies_2}
  \omega_{BD} = \frac{H^2}{8 \uppi^2} (1-Z)^{-1} + \frac{m^2}{16 \uppi^2}\, \FF\left(\frac{3}{2} + \nu, \frac{3}{2} - \nu; 2; \frac{1-Z}{2}\right) \ln \frac{1-Z}{2} + \frac{m^2}{16 \uppi^2} \widetilde W
\end{equation*}
with a certain smooth function $\widetilde W = (\widetilde W \circ Z)(x_1, x_2)$.

In the conformally flat patch of the de Sitter spacetime $Z$ takes a simple form and can be expressed in terms of the squared, signed geodesic distance of the conformally related Minkowski space as
\begin{equation*}
  Z(x_1,x_2) = 1 - \frac{\sigma_\MM(x_1, x_2)}{2 \tau_1 \tau_2},
  \quad
  \sigma_\MM(x_1,x_2) \defn - (\tau_1 - \tau_2)^2 + (\xx_1 - \xx_2 )^2.
\end{equation*}
On $(M, \overline g)$ we therefore obtain an expression for $\omega_{BD}$ which resembles the Hadamard form on the conformally related Minkowski space:
\begin{equation}
  \label{eq:bunch_davies_expansion}
  \omega_{BD}(x_1, x_2) = \lim_{\varepsilon \to 0^+} \frac{H^2}{4 \uppi^2} \frac{\tau_1 \tau_2}{\sigma_\MM(x_1, x_2) + 2 \im \varepsilon\, (\tau_1 - \tau_2) + \varepsilon^2} + \text{less singular terms}.
\end{equation}
Notice that the less singular contributions vanish in the limit of zero mass.

As can be seen in \eqref{eq:einstein_prob_trace_2}, we need to compute the square of the two-point distribution of the state in question, \ie, in this case the square of $\omega_{BD}$.
For our purposes it will be sufficient to compute the square of the leading singularity in the Hadamard state.
To this end we notice that, up to a trivial rescaling by $H^2 \tau_1 \tau_2$, the leading singularity coincides with the vacuum state $\omega_\MM$ for a massless field theory on the conformally related Minkowski spacetime.
The square of the massless two-point distribution on Minkowski space is
\begin{equation*}
  \omega_\MM(x_1, x_2)^2 = \lim_{\varepsilon\to 0^+} \left( \frac{1}{4 \uppi^2} \frac{1}{\sigma_\MM(x_1, x_2) + 2 \im \varepsilon\, (\tau_1 - \tau_2) + \varepsilon^2} \right)^2.
\end{equation*}
Writing $\omega_\MM$ in terms of its spatial Fourier transform, an expression for the spatial Fourier transform of the square of the massless Minkowski vacuum can be obtained as
\begin{equation}
  \label{eq:spectrum_omega2}
  \omega_\MM(x_1,x_2)^2 = \lim_{\varepsilon \to 0^+} \frac{1}{128 \uppi^5}
  \int_{\RR^3} \e^{\im \kk \cdot (\xx_1 - \xx_2)}
  \int_{k}^\infty \e^{-\im p\, (\tau_1 - \tau_2)}\, \e^{- \varepsilon p }\, \dif p\, \dif^3\!\kk,
\end{equation}
where $x_i \defn (\tau_i,\xx_i)$ and $k \defn \abs{\kk}$.
Later on we will use this expression in order to obtain the power spectrum of $\Psi$.

\subsection{Power spectrum of the metric perturbations}

We want to compute the power spectrum $P(\tau, \kk)$ of the two-point correlation of $\Psi$ at the time $\tau$.
Since both the spacetime and the chosen state are invariant under spatial translation, it can be defined as
\begin{equation*}
  \E[\big]{\Psi(\tau, \xx_1)\, \Psi(\tau, \xx_2)} \defn \frac{1}{(2 \uppi)^3} \int_{\RR^3} P(\tau,\kk)\, \e^{\im \kk \cdot (\xx_1 - \xx_2)}\, \dif^3\!\kk.
\end{equation*}
To obtain $P$, we first need an expression for the retarded operator $\Delta_R$ corresponding to \eqref{eq:hyp_op}:
\begin{align*}
  (\Delta_R\,f)(\tau, \xx) & = \frac{1}{(2 \uppi)^3} \int_{\RR^3} \int_{-\infty}^{\tau}
  \widehat{\Delta}_R(\tau, \tau_1, \kk)\,
  \widehat{f}(\tau_1, \kk)\, \e^{\im \kk \cdot \xx}\, \dif\tau_1\, \dif^3\!\kk,
  \quad \text{with}
  \\
  \widehat{\Delta}_R(\tau, \tau_1, \kk) & \defn
  - \frac{1}{6 H^2} \frac{\tau^2}{\tau_1^4} \frac{\sqrt{3}}{k} \sin\left(k\, (\tau - \tau_1)/\sqrt{3}\right),
\end{align*}
where $f$ is a compactly supported, smooth function, $\,\widehat{\cdot}\,$ is the spatial Fourier transform.\footnote{The conventions for the Fourier transform observed here are:
\begin{equation*}
  \widehat{f}(k) = \int_{\RR^3} f(x)\, \e^{-\im \kk \cdot \xx}\, \dif^3\!\xx,
  \quad
  f(x) = (2 \uppi)^{-3} \int_{\RR^3} \widehat{f}(k)\, \e^{\im \kk \cdot \xx}\, \dif^3\!\kk.
\end{equation*}}
We can then rewrite \eqref{eq:retarded_applied_to_w} in Fourier space to obtain
\begin{equation*}
  P(\tau, \kk) = 2 m^4 \int_{-\infty}^\tau \int_{-\infty}^\tau
  \widehat{\Delta}_R(\tau, \tau_1, \kk)\, \widehat{\Delta}_R(\tau, \tau_2, \kk)\,
  \widehat{\omega^2_{BD}}(\tau_1, \tau_2, \kk)\,
  \dif\tau_1\, \dif\tau_2.
\end{equation*}
Note that the symmetrization of the state is taken care of indirectly by the equal limits of the two integrations.

As discussed above (see \eqref{eq:bunch_davies_expansion} and the following paragraph), we will compute the contribution due to the leading singularity of the Hadamard state:
\begin{equation*}
  P_0(\tau, \kk) \defn 2 H^4 m^4 \int_{-\infty}^\tau \int_{-\infty}^\tau
  \widehat{\Delta}_R(\tau, \tau_1, \kk)\, \widehat{\Delta}_R(\tau, \tau_2, \kk)\,
  \tau_1^2\, \tau_2^2\, \widehat{\omega^2_\MM}(\tau_1, \tau_2, \kk)\,
  \dif\tau_1\, \dif\tau_2.
\end{equation*}
We emphasize at this point that, because of the form of \eqref{eq:spectrum_omega2} and of $\widehat{\Delta}_R$, no $\kk-$infrared singularity appears in $P_0(\tau,\kk)$ at finite $\tau$.
Recall also that the error we are committing, using $P_0(\tau,\kk)$ at the place of $P(\tau,\kk)$, tends to vanish in the limit of small masses.
Inserting the spectrum of $\omega_\MM^2$ obtained in \eqref{eq:spectrum_omega2} and switching the order in which the integrals are taken (for $\varepsilon > 0$), we can write
\begin{equation}
  \label{eq:P0_with_A}
  P_0(\tau, \kk) =
  \lim_{\varepsilon \to 0^+} \frac{m^4}{16 \uppi^2}
  \int_{k}^\infty
  \frac{1}{k^4}\, \abs*{A\left(\tau, k/\sqrt{3}, p\right)}^2\, \e^{-\varepsilon p}\, \dif p,
\end{equation}
where we have introduced the auxiliary function
\begin{equation}
  \label{eq:defA}
  A(\tau, \kappa, p) \defn \int_{-\infty}^\tau
  \frac{\kappa\, \tau^2}{\tau_1^2}
  \sin\big( \kappa\, (\tau - \tau_1) \big)\,
  \e^{-\im p \tau_1}\,
  \dif\tau_1,
\end{equation}
which can also be written in closed form in terms of the generalized exponential integral\footnote{For a definition and various properties of these special functions see \eg \cite[Chap. 8]{olver:2010}.} $E_2$ as
\begin{equation*}
  \label{eq:A_with_E2}
  A(\tau, \kappa, p) =
  A(\kappa\, \tau, p\, \tau) =
  \frac{\im}{2}\, \kappa\, \tau\,
  \Big(
    E_2 \big( \im\, (p + \kappa)\, \tau \big)\, \e^{ \im \kappa \tau} -
    E_2 \big( \im\, (p - \kappa)\, \tau \big)\, \e^{-\im \kappa \tau}
  \Big)
\end{equation*}
for $p \geq \kappa > 0$ and by the complex conjugate of this expression if $\kappa > p, \kappa > 0$.
In the following study of the form of the power spectrum $P_0$ the auxiliary function $A$ will be instrumental.

\begin{lemma}
  \label{lem:A}
  For $\abs{p} \neq \kappa > 0$, $A(\tau, \kappa, p)$ has the $\tau$-uniform bound
  \begin{equation}
    \label{eq:bound_A}
    \abs{A} \leq \frac{4\, \kappa^2}{\abs{\kappa^2 - p^2}}.
  \end{equation}
  For large negative times it satisfies the limit
  \begin{equation}
    \label{eq:limit_A}
    \lim_{\tau \to -\infty} \abs{A} = \frac{\kappa^2}{\abs{\kappa^2 - p^2}}.
  \end{equation}
\end{lemma}
\begin{proof}
  Using the fact that
  \begin{equation*}
    \e^{-\im p \tau_1} = \left( \od[2]{}{\tau_1} + \kappa^2 \right) \frac{\e^{-\im p \tau_1}}{\kappa^2 - p^2},
  \end{equation*}
  we can perform two integrations by parts to obtain
  \begin{align*}
    A(\tau, \kappa, p) & = \frac{\kappa^2}{\kappa^2 - p^2} \big( \e^{-\im p \tau} + R(\tau, \kappa, p) \big),
    \quad \text{with}
    \\
    R(\tau, \kappa, p) & \defn \tau^2 \int_{-\infty}^\tau
    \left(
      \frac{4}{\tau_1^3} \cos\big(\kappa\, (\tau - \tau_1)\big)
      +
      \frac{6}{\kappa\, \tau_1^4} \sin\big(\kappa\, (\tau - \tau_1)\big)
    \right)\, \e^{-\im p \tau_1}\, \dif\tau_1.
  \end{align*}
  It is now easy to obtain an upper bound for $R$ which is uniform in conformal time, namely $\abs{R} \leq 3$, which then yields the bound \eqref{eq:bound_A}.

  For the second part of the proposition we perform a change of the integration variable to $x = \tau_1 / \tau$:
  \begin{equation*}
    R(\tau, \kappa, p) = - \int_1^{\infty}
    \left(
      \frac{4}{x^3} \cos\big(\kappa\, \tau\, (1 - x)\big)
      +
      \frac{6}{\kappa\, \tau\, x^4} \sin\big(\kappa\, \tau\, (1 - x)\big)
    \right)\, \e^{-\im p \tau x}\, \dif x.
  \end{equation*}
  The contribution proportional to $1/\tau$ in $R$ is bounded by $C(\kappa)/\abs{\tau}$ and thus vanishes in the limit $\tau \to -\infty$.
  Moreover, since $\abs{p} \neq \kappa$ and ${1}/{x^3}$ is $L^1$ on $[1,\infty)$, we can apply the Riemann--Lebesgue lemma and see that this contribution vanishes in the limit $\tau \to -\infty$.
  The remaining part of $\abs{A}$ is $\kappa^2 \abs{\kappa^2 - p^2}^{-1}$, which is independent of $\tau$, and thus the limit \eqref{eq:limit_A} holds true.
\end{proof}

Note that the bound for $A$ obtained above is not optimal.
Numerical integration indicates that $\abs{A}^2$ is monotonically decreasing in $\tau$ and thus bounded by the limit stated in \eqref{eq:limit_A} (see also Fig. \ref{fig:P0_plot}).
Nevertheless, we can use this lemma to derive the following bounds and limits for $P_0$:

\begin{proposition}
  \label{prop:bound_large_limit}
  The leading contribution $P_0$ to the power spectrum of the potential $\Psi$ induced by a conformally coupled massive scalar field in the Bunch--Davis state is bounded by the Harrison--Zel'dovich spectrum uniformly in time, namely
  \begin{equation*}
    \abs[\big]{P_0(\tau, \kk)} \leq \frac{16\, C}{\abs{\kk}^3},
    \quad
    C \defn \frac{3 - 2 \sqrt{3}\, \arccoth\!\sqrt{3}}{192 \uppi^2}\, m^4,
  \end{equation*}
  and it tends to the Harrison--Zel'dovich spectrum for $\tau \to -\infty$, \ie,
  \begin{equation*}
    \label{eq:limit_P0}
    \lim_{\tau \to -\infty} P_0(\tau, \kk) = \frac{C}{\abs{\kk}^3}.
  \end{equation*}
\end{proposition}
\begin{proof}
  The proof can be easily obtained using the $\tau$-uniform estimate \eqref{eq:bound_A} obtained in lemma \ref{lem:A} and computing the integral
  \begin{equation*}
    \abs[\big]{P_0(\tau, \kk)}
    \leq \frac{m^4}{\uppi^2} \int_k^\infty \left(\frac{1}{3 p^2 - k^2}\right)^2\, \dif p
    = \frac{3 - 2 \sqrt{3}\, \arccoth\!\sqrt{3}}{12 \uppi^2} \frac{m^4}{k^3}.
  \end{equation*}

  Having shown the first part of the proposition, let us now analyze the limit
  \begin{equation*}
    \lim_{\tau \to -\infty} P_0(\tau, \kk) = \frac{m^4}{16 \uppi^2} \int_k^\infty \frac{1}{k^4}
    \lim_{\tau \to -\infty} \abs*{A\left(\tau, k/\sqrt{3}, p\right)}^2\, \dif p,
  \end{equation*}
  where we have taken the $\tau$-limit before the integral and already evaluated the $\varepsilon$-limit because $\abs{A}^2$ is bounded by an integrable function uniformly in time.
  Inserting the limit \eqref{eq:limit_A} from lemma \ref{lem:A}, we can compute the $p$-integral
  \begin{equation*}
    \lim_{\tau \to -\infty} P_0(\tau, \kk)
    = \frac{m^4}{16 \uppi^2} \int_k^\infty (3 p^2 - k^2)^{-2}\, \dif p
    = \frac{3 - 2 \sqrt{3}\, \arccoth\!\sqrt{3}}{192 \uppi^2} \frac{m^4}{k^3},
  \end{equation*}
  thus concluding the proof.
\end{proof}

We can complement the results of proposition \ref{prop:bound_large_limit} with the following observation:
\begin{proposition}
  \label{prop:P0_k3}
  The power spectrum $P_0$ has the form
  \begin{equation*}
    P_0(\tau, \kk) = \frac{\mathcal{P}_0(\abs{\kk} \tau)}{\abs{\kk}^3},
  \end{equation*}
  where $\mathcal{P}_0$ is a function of $\abs{\kk} \tau$ only.
\end{proposition}
\begin{proof}
  Noting that $A(\tau, \kappa, p)$ is a function of $\kappa\, \tau$ and $p\, \tau$ only and performing the $\varepsilon$-limit in \eqref{eq:P0_with_A} inside the integral, this can be seen by the substitution $x = p\, \tau$ in \eqref{eq:P0_with_A}.
\end{proof}

We would like to improve the estimate of $P_0(\tau, \kk)$ for $\tau$ close to zero.
Adhering to our previous strategy, we shall first give a new estimate for $A(\tau, k, p)$:

\begin{lemma}
  The auxiliary function $A(\tau, \kappa, p)$ is bounded by
  \begin{equation*}
    \abs[\big]{A(\tau, \kappa, p)} \leq - 2\, \frac{\kappa^2\, \tau}{\abs{p}}, \quad p \neq 0,\quad \tau<0.
  \end{equation*}
\end{lemma}
\begin{proof}
  Recalling the form of $A$ given in \eqref{eq:defA} and integrating by parts, where we use that $\e^{-\im p \tau_1} = \im\, p^{-1}\, \partial_{\tau_1} \e^{-\im p \tau_1}$, we find
  \begin{equation*}
    A(\tau, \kappa, p) = \frac{\im\, \kappa^2 \tau^2}{p} \int_{-\infty}^\tau
    \left(
      \frac{1}{\tau_1^2} \cos\big(\kappa\, (\tau - \tau_1)\big) +
      \frac{2}{\kappa\, \tau_1^3} \sin\big(\kappa\, (\tau - \tau_1)\big)
    \right)\, \e^{-\im p \tau_1}\, \dif\tau_1.
  \end{equation*}
  We then take the absolute value and estimate the trigonometric functions, which gives us a bound on $A$, namely
  \begin{equation*}
    \abs[\big]{A(\tau, \kappa, p)}
    \leq \frac{\kappa^2 \tau^2}{\abs{p}} \int_{-\infty}^\tau
    \left( \frac{1}{\tau_1^2} - 2\, \frac{\tau - \tau_1}{\tau_1^3} \right)\, \dif\tau_1
    = - 2\, \frac{\kappa^2\, \tau}{\abs{p}}. \qedhere
  \end{equation*}
\end{proof}

Performing the integration in $p$ analogously to the second part of proposition \eqref{prop:bound_large_limit}, the last lemma immediately leads to a corresponding bound for $P_0$:

\begin{proposition}
  \label{prop:small_limit}
  The leading contribution $P_0$ of the power spectrum of the potential $\Psi$ satisfies the inequality
  \begin{equation*}
    \abs[\big]{P_0(\tau, \kk)} \leq \frac{m^4}{36 \uppi^2} \frac{\tau^2}{\abs{\kk}}
  \end{equation*}
  and therefore, in particular, $P_0(0, \kk) = 0$.
\end{proposition}

The rescaled power spectrum $\mathcal{P}_0(\abs{\kk} \tau)$ can be analyzed numerically and a plot is shown in Fig. \ref{fig:P0_plot}.
It clearly exhibits the asymptotic behaviour of $P_0$ discussed in propositions \ref{prop:bound_large_limit} and \ref{prop:small_limit}.
Note that the horizontal axis is logarithmically scaled to highlight the behavior of $\mathcal{P}_0$ for small $\abs{\kk} \tau$, which would be concealed by the fast approach of $\mathcal{P}_0$ to its bound had we used a linear scaling.

\begin{figure}
  \centering
	\includegraphics{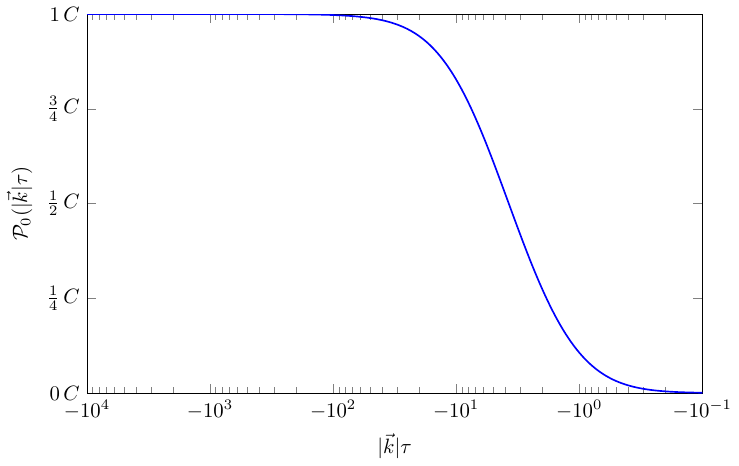}
  \caption{Logarithmic plot of the rescaled power spectrum $\mathcal{P}_0(\abs{\kk} \tau)$, where $C$ is the same proportionality constant as in proposition \ref{prop:bound_large_limit}.}
  \label{fig:P0_plot}
\end{figure}

In this section we have used the leading singularity of the two-point function of the Bunch--Davis state on a de Sitter universe to compute the influence of quantum matter on the power spectrum of the metric perturbation $\Psi$.
(Recall that considering only the leading singularity in the Bunch--Davis state also corresponds to the limit of vanishing mass.)
We have seen that this results in an almost scale-invariant power spectrum.
We stress that such a singularity is not a special feature of the Bunch--Davis state but is common for every Hadamard state.
Moreover, although our analysis has been done on a de Sitter universe, similar quantum states have been constructed on universes which are asymptotically de Sitter spaces in the past \cite{dappiaggi:2009a,dappiaggi:2009b}.
All these states tend to the Bunch--Davis state for $\tau \to -\infty$ and are of Hadamard form.


\subsection{Non-Gaussianities of the metric perturbations}

It follows from \eqref{eq:einstein_prob_trace_n} that the $n$-point correlation for $\Psi$ will, in general, not vanish.
Also for odd $n$ they will be different from zero and hence $\Psi$ is not a Gaussian random field.
As a first measure of the non-Gaussianity of $\Psi$ one usually calculates its three-point correlation function or the corresponding bispectrum $B$:
\begin{multline*}
  \E[\big]{\Psi(\tau, \xx_1)\, \Psi(\tau, \xx_2)\, \Psi(\tau, \xx_3)} \defn
  \frac{1}{(2 \uppi)^9} \iiint_{\RR^9}
  \delta(\kk_1 + \kk_2 + \kk_3)\, B(\tau, \kk_1, \kk_2, \kk_3)
  \\ \times
  \e^{\im\, (\kk_1 \cdot \xx_1 + \kk_2 \cdot \xx_2 + \kk_3 \cdot \xx_3)}\,
  \dif^3\!\kk_1\, \dif^3\!\kk_2\, \dif^3\!\kk_3.
\end{multline*}

Assuming non-zero $\kk_1$, $\kk_2$ and $\kk_3$, we will derive the form of the bispectrum $B$ considering (as above) only the contribution due to the leading singularity of the Bunch--Davis state, which we will denote by $B_0$.
We will follow the same steps that lead us to the calculation of the power spectrum in the previous section.
That is, we apply the retarded propagator $\Delta_R$ of \eqref{eq:hyp_op} as in \eqref{eq:retarded_applied_to_w} to the right-hand side of \eqref{eq:einstein_prob_trace_n} for $n=3$ to obtain an equation for $\Psi$ and insert for the two-point distribution the conformally rescaled two-point distribution of the massless Minkowski vacuum.
The result can again be expressed in terms of the auxiliary function $A$ defined in \eqref{eq:defA}:
\begin{multline}
  \label{eq:B0_with_A}
  B_0(\tau, \kk_1, \kk_2, \kk_3) = \lim_{\varepsilon \to 0^+}
  \frac{m^6}{32 \sqrt{3}\, k_1^2\, k_2^2\, k_3^2} \int_{\RR^3}
  \left(
    \frac{\e^{-\varepsilon\, (\omega_\pp(-\kk_1) + \omega_\pp(\kk_3) + \abs{\pp})}}{\omega_\pp(-\kk_1)\, \omega_\pp(\kk_3)\, \abs{\pp}}\,
    A\left(\tau, \frac{\abs{\kk_1}}{\sqrt{3}},   \omega_\pp(-\kk_1) + \abs{\pp} \right) \right.
    \\ \times \left.
    A\left(\tau, \frac{\abs{\kk_3}}{\sqrt{3}}, - \omega_\pp(\kk_3) - \abs{\pp} \right)\,
    A\left(\tau, \frac{\abs{\kk_2}}{\sqrt{3}},   \omega_\pp(\kk_3) - \omega_\pp(-\kk_1) \right)
    + \text{permutations}
  \right)\,
  \dif^3\!\pp,
\end{multline}
where $\omega_\pp(\kk) \defn \abs{\kk + \pp}$ and the sum is over all permutations of $1, 2, 3$.

We can apply the same bound on $A$ which has been used in the previous section to bound the power spectrum $P_0$ to produce a bound on the integrand of $B_0$ almost everywhere.\footnote{We cannot bound the integrand of $B_0$ everywhere using \eqref{eq:bound_A} because $\abs{\kk_2} / \sqrt{3} \neq \abs[\big]{\omega_\pp(\kk_3) - \omega_\pp(-\kk_1)}$ (and permutations) does not hold everywhere.}
Nevertheless, the singularity in the integrand in \eqref{eq:B0_with_A} is integrable, \ie, $B_0$ is bounded.
As a consequence we can perform the limit $\varepsilon \to 0^+$ inside the integral.

\begin{proposition}
  The leading contribution $B_0$ of the bispectrum of the metric perturbation $\Psi$ has the form
  \begin{equation*}
    B_0(\tau, \kk_1, \kk_2, \kk_3) = \frac{\mathcal{B}_0(k_1 \tau, k_2 \tau, k_3 \tau)}{k_1^2\, k_2^2\, k_3^2},
  \end{equation*}
  where $\mathcal{B}_0$ is a function of $k_1 \tau$, $k_2 \tau$, and $k_3 \tau$ only and $k_i = \abs{\kk_i}$.
\end{proposition}
\begin{proof}
  Analogously to proposition \ref{prop:P0_k3}, we note that after a change of variables $\vec{x} = \tau\, \vec{p} $ the integrand in \eqref{eq:B0_with_A} is a function of $k_1 \tau$, $k_2 \tau$, and $k_3 \tau$ only.
\end{proof}

To finish our discussion about non-Gaussianities, we notice that, although the employed quantum field is a linear one, we obtained a three-point function for $\Psi$ which is similar to the one obtained by Maldacena \cite{maldacena:2003} who has quantized metric perturbations outside the linear approximation.


\section{Conclusions}
\label{sec:conclusions}

In this paper we have analyzed the influence of quantum matter fluctuations on metric perturbations over de Sitter backgrounds.
We used techniques proper of quantum field theory on curved spacetime to regularize the stress-energy tensor and to compute its fluctuations.
In particular, we interpreted the perturbations of the curvature tensors as the realization of a stochastic field.
We then obtained the $n$-point functions of such a stochastic field as induced by the $n$-point functions of a quantum stress tensor by means of semiclassical Einstein equations.

The proposed approach bears some resemblance to stochastic gravity but considers also higher moments of the stress-energy tensor.
However, contrary to \eg stochastic gravity, the formalism is presented in an \emph{ad hoc} fashion and is not formally derived from an eventual quantum gravity.
Nevertheless, we saw that a reasonable power spectrum can be obtained.

We also noticed that, while the expectation value of the stress-energy tensor is characterized by renormalization ambiguities, this is no longer the case when fluctuations are considered.
Hence the obtained results are independent on the particular regularization used to define the stress tensor.

In order to keep superficial contact with literature on inflation, we investigated perturbations of the scalar curvature generated by a Newtonian metric perturbation, which is related to the standard Bardeen potentials.
However, the considered model is certainly oversimplified to cover any real situation and is not gauge invariant.

Within this model it was possible to recover an almost-Harrison--Zel'dovich power spectrum for the considered metric perturbation.
Furthermore, the amplitude of such a power spectrum depends on the field mass which is a free parameter in our model and can be fixed independently of $H$.
At the same time, since it does not depend on the Hubble parameter of the background metric, this indicates that it is not a special feature of de Sitter space.
At least close to the initial singularity, the obtained result depends only on the form of the most singular part of the two-point function of the considered Bunch--Davis state.
We thus argue that a similar feature is present in every Hadamard state.
Furthermore, as conjectured in an earlier draft of this paper, it has been recently proven in \cite{dappiaggi:2014}, that similar effects holds also for backgrounds which are only asymptotically de Sitter in the past.

Finally we notice that, since the stress-energy tensor is not linear in the field, its probability distribution cannot be of Gaussian nature.
Thus we showed that non-Gaussianities arise naturally in this picture.

\begin{acknowledgments}
  We would like to thank C. Dappiaggi, K. Fredenhagen, T.-P. Hack and V. Moretti for helpful discussions. The work of N.P. has been supported partly by the Indam-GNFM project ``Effetti topologici e struttura della teoria di campo interagente''.
\end{acknowledgments}


\end{document}